\documentclass[12pt]{article}
\usepackage{amsmath,amssymb,amsthm,bm}
\usepackage{mathbbol,bbm}
\usepackage[final]{showkeys}
\usepackage{fancyhdr}
\usepackage{verbatim}

\newcommand{\Cx}{{\mathbb C}}
\newcommand{\Ir}{\mathbb{Z}}
\newcommand{\Nl}{\mathbb{N}}

\newcommand{\Rl}{{\mathbb R}}

\newcommand{\Ex}{\mathbb{E}}

\newcommand{\idty}{\mathbb{1}}
\DeclareMathOperator{\id}{id}

\DeclareMathOperator*{\tr}{Tr}
\newcommand{\<}{\langle}
\renewcommand{\>}{\rangle}
\providecommand{\abs}[1]{\lvert#1\rvert}
\providecommand{\norm}[1]{\lVert#1\rVert}

\renewcommand{\c}[1]{\mathcal{#1}}
\newcommand{\g}[1]{\mathfrak{#1}}
\newcommand{\s}[1]{\mathsf{#1}}
\renewcommand{\r}[1]{\mathrm{#1}}

\newtheorem{lemma}{Lemma}
\newtheorem{proposition}{Proposition}
\newtheorem{thm}{Theorem}
\theoremstyle{definition}
\newtheorem{construction}{Construction}

\newtheorem{example}{Example}

\DeclareMathOperator*{\loplus}{\mbox{\Large\mbox{$\oplus$}}}
\DeclareMathOperator*{\lotimes}{\mbox{\Large\mbox{$\otimes$}}}

\newcommand{\hl}{}

\begin{document}
\begin{center}
{\LARGE Fermionic Markov Chains} \\[12pt]
M.~Fannes and J.~Wouters\footnote{Current affiliation: Meteorologisches Institut, University of Hamburg}\\[6pt]
\texttt{mark.fannes@fys.kuleuven.be} \\
\texttt{jeroen.wouters@zmaw.de}\\
\vspace{5pt}
Instituut voor Theoretische Fysica \\
K.U.Leuven, Belgium \\
\end{center}

\abstract{We study a quantum process that can be considered as a quantum analogue for the classical Markov process. We specifically construct a version of these processes for free Fermions. For such free Fermionic processes we calculate the entropy density. This can be done either directly using Szeg\"o's theorem for asymptotic densities of functions of Toeplitz matrices, or through an extension of said theorem to rates of functions, which we present in this article.} 

\section{Introduction}
\label{sec:intro}
Quantum channels describe the black box dynamics of small open quantum systems, i.e.\ a quantum system evolving in contact with an inaccessible environment. Technically, a channel is a completely positive map, which maps an input density matrix into an output density matrix. It corresponds to a one-shot random evolution of the system. A classical channel is a stochastic matrix.

In this article, we consider the construction of a quantum process associated to a channel, much like a stochastic matrix generates a Markov process. Introducing \hl{ in a quantum system} multi-time correlations that are compatible with a given channel is, however, much more delicate than in the classical context. Such an amplification to a process is not generally possible and, if possible, the process is not unique. We focus on entropic properties of such processes, both for determining extensions with minimal entropy and obtaining a measure of randomness in the given channel.

The construction uses a generalization of matrix product states which was introduced under the name of finitely correlated states, see~\cite{Accardi81,Fannes89}. While matrix product states prove to be a very useful class of pure states on quantum spin chains~\cite{Perez07,Klumper91}, suitable for studying ground state properties, one needs to go beyond such states in the context of channels \hl{due to the mixing of pure states by generic channels}. In fact, even in the classical context, the class of processes we consider includes hidden Markov processes.

In the study of the information carrying capacity of classical channels with memory, the entropy density of hidden Markov processes also arises~\cite{Wouters09}. The entropy density of finitely correlated states is expected to play a role in the information capacity of quantum channels with memory.

In this article, however, we consider a simpler problem and introduce a construction of free Fermionic Markov processes compatible with a free Fermionic channel. We show in particular that, \hl{instead of through a direct calculation}, the entropy density can also be obtained as the asymptotic entropy production under the shift dynamics.

The structure of the article is as follows. In Section~\ref{sec:class} we review properties of classical Markov and hidden Markov processes. We pay particular attention to a method to calculate the entropy density, based on a De~L'H\^opital-like property of strongly subadditive functions. In Section~\ref{sec:quant} we introduce the quantum version of hidden Markov processes. In Section~\ref{sec:fermionic} we consider a free Fermionic version of such processes. We then turn to the main result in Section~\ref{sec:density}. We show that the De~L'H\^opital-like property, which connects averages of functions to their growth rates, can be extended to a much wider class than the strongly subadditive functions. As the density matrices of free Fermionic systems are basically Toeplitz matrices, this amounts to an extension Szeg\"o's theorem for averages of functions to rates.

\section{The classical case}
\label{sec:class}

We first consider classical Markov and hidden Markov processes. Let $P$ be a stochastic matrix over a finite state space $\Omega = \{ 1, 2, \ldots, d \}$: the entries of the $d \times d$ matrix $P$ are non-negative and the row sums are equal to 1. The entry $P_{\omega_1 \omega_2}$ specifies the jump probability from state $\omega_1$ to state $\omega_2$ and hence $P$ defines a stochastic dynamics in discrete time. Generically $P$ has a non-degenerate eigenvalue 1 and so has its transpose $P^{\textsf T}$. The Perron-Frobenius theorem asserts that the absolute values of the eigenvalues are not larger than 1 and that the entries of the eigenvector of $P^{\textsf T}$ corresponding to the eigenvalue 1 can be chosen non-negative. A proper normalization provides us therefore with a probability vector $\mu$ over $\Omega$ such that $P^{\textsf T}\mu = \mu$. For a generic $P$, $\mu$ is faithful and we have exponentially fast convergence to the invariant measure: there exist $C \ge 0$ and $0 \le \gamma < 1$ such that
\begin{equation}
\bigl\Vert ( P^{\textsf T})^n \nu - \mu \bigr\Vert_1 \le C \gamma^n,\enskip n \in \Nl,\ \text{$\nu$ probability vector}.
\end{equation} 

Moreover, $P$ generates a natural stationary stochastic process by assigning to a path $(\omega_0, \omega_1, \ldots, \omega_n)$ the probability
\begin{equation}
\bigl\< \omega_0, \omega_1, \ldots, \omega_n \bigr\>_{[0,n]} = \mu_{\omega_0} P_{\omega_0\, \omega_1} P_{\omega_1\, \omega_2} \cdots P_{\omega_{n-1}\, \omega_n},\enskip \omega_j \in \Omega,\ n = 0, 1, \ldots
\label{mc}
\end{equation} 
This is a one step Markov process: the probability for reaching the state $\omega_n$ at time $n$ given the full history $(\omega_0, \omega_1, \ldots, \omega_{n-1})$ is the same as that for reaching $\omega_n$ starting at time $n-1$ at $\omega_{n-1}$.

The entropy production or rather the mean entropy of this Markov chain is a natural way to quantify the randomness of $P$
\begin{equation}
\s h(P) := \lim_{n \to \infty} \frac{1}{n}\, \s H\bigl( \<\ \>_{[0,n-1]} \bigr)
\label{mc:ent}
\end{equation}
Here, $\s H$ is the usual Shannon entropy of a probability vector. In fact, in the event of multiple stationary measures for $P$ the entropy also depends on the chosen initial measure.

For strongly subadditive function, such as the Shannon entropy, the following discrete version of De~L'H\^opital's rule can be proven\cite{Alicki01,Kay01}:
\begin{equation}
\lim_{n \to \infty} \frac{1}{n}\, \s H\bigl( \<\ \>_{[0,n-1]} \bigr) = \lim_{n \to \infty} \Bigl( \s H(\<\ \>_{[0,n]}) - \s H(\<\ \>_{[0,n-1]}) \Bigr).
\label{be}
\end{equation}
Using this equality a simple calculation shows that the entropy density of the Markov process (\ref{mc}) is given by
\begin{equation}
\s h(P) = \< \s H_\text{trans} \>_\mu \,, \label{eq:markov_density}
\end{equation}
where $\s H_\text{trans} (\omega)$ is the entropy of the conditional probabilities related to the transition from $\omega$ to the next state:
\begin{equation}
\s H_\text{trans} (\omega) = - \sum_\sigma P_{\omega \sigma} \log P_{\omega \sigma} \label{class_density} \,.
\end{equation}


This construction is however difficult to carry over to quantum systems due to the prominent role of paths. The probabilistic nature of quantum mechanics and the non-uniqueness of the choice of basis for the Hilbert space of the system make paths an unnatural concept in quantum mechanics. A construction that is better suited for generalization is based on positive maps.
We consider stochastic matrices with $d$ rows and $d^2$ columns. An observable on the discrete state space $\Omega$ can be seen as a vector $f \in \Rl^d$ and we use the notation $\bm 1$ for the constant function 1, i.e.\ every entry of $\bm 1$ is equal to 1. 

\begin{construction}
\label{con2}
Let $P$ be a stochastic $d \times d$ matrix with invariant measure $\mu$ and let $Q$ be a $d \times d^2$ stochastic matrix that satisfies the compatibility condition
\begin{equation}
Q\, (f\otimes \bm 1) = Q\, (\bm 1 \otimes f) = P\, f,\enskip f \in \Rl^d.
\label{cmc}
\end{equation}
Any such $Q$ defines a stationary measure on the half-chain $\times^{\Nl} \Omega$ with marginals
\begin{equation}
\bigl\< f_n \bigr\> = \mu\Bigl( Q \bigl( Q \otimes \idty) \cdots \bigl( Q \otimes \idty \otimes \cdots \otimes \idty \bigr) \bigl( \bm 1 \otimes f_n \bigr) \Bigr),\enskip f_n \in \lotimes_0^{n-1} \Rl^d.
\label{cfcs}
\end{equation}
\begin{flushright}
$\square$
\end{flushright}
\end{construction}

The set of stochastic matrices $Q$ obeying~(\ref{cmc}) is closed, convex and non-empty. E.g.\ the Markov chain~(\ref{mc}) is obtained by choosing
\begin{equation}
Q(f \otimes g) = P(fg),\enskip f,g \in \Rl^d,
\label{marext}
\end{equation}
where $fg$ is the entrywise product of $f$ and $g$. For a general $Q$, the measure~(\ref{cfcs}) is a stationary hidden Markov process.

The mean entropy of a hidden Markov process can be computed using a method due to Blackwell, see~\cite{Fannes92} and~\cite{Blackwell57}. This computation is based on the discrete version of De~L'H\^opital's rule for the mean entropy given in Equation~\ref{be}. The relation between the $n$ and $n-1$ site marginals of a hidden Markov process is given by a transfer matrix like relation, as seen from~(\ref{cfcs}). This allows to express the mean entropy as an average of entropies of probability vectors over $\Omega$
\begin{equation}
\lim_{n \to \infty} \frac{1}{n}\, \s H\bigl( \<\ \>_n \bigr) = \int \!\varphi(d\nu)\, \s H^ \prime_\text{trans}(\nu).
\end{equation}  
In this formula, $\nu$ is varies over the set of probability vectors over $\Omega$ and $\varphi$ is a measure on this set of probability vectors. The measure $\varphi$ is the unique stationary measure of a dynamical system on probability vectors over $\Omega$ that is determined by $Q$.

This formula bears some similarity to Eq.~\ref{eq:markov_density}. The measure $\varphi$ weighs the possible past paths, much as the measure $\mu$ in Eq.~\ref{eq:markov_density} weighs the possible configuration at the previous time step. As the hidden Markov process also has correlations with previous time steps, this weighting term gets significantly more complicated. The entropy term $\s H^ \prime_\text{trans}(\nu)$ is an entropy function related to the transition from one state to the next, much as the function $\s H_\text{trans}(\nu)$ in Eq.~\ref{class_density} is determined by the transition probabilities $P_{\omega \sigma}$.

Numerical experiments suggest that among the extensions that satisfy the compatibility condition~(\ref{qmc}), the Markov chain extension~(\ref{marext}) has the smallest entropy.

\section{The quantum case}
\label{sec:quant}

The natural quantum analogue of a stochastic matrix is a unity preserving completely positive (CP) map $\Gamma$ acting on the $d$-dimensional complex matrices $\c M_d$. Such maps send a pure state into a mixed one and are therefore stochastic. Generically $\Gamma$ has a non-degenerate eigenvalue 1, the corresponding eigenvector of the transpose, mostly called the dual, is a faithful density matrix $\rho$ and exponentially fast return to equilibrium holds: for any initial density matrix $\sigma$
\begin{equation}
\bigl\Vert \sigma \circ \Gamma^n - \rho \bigr\Vert_1 \le C \gamma^n,\enskip n \in \Nl,\ \text{$\sigma$ density matrix}.
\end{equation} 

Unlike the classical case, where there is a trivial connection between stochastic matrices and Markov processes, there is no straightforward extension to a process. A first reason is that a general density matrix admits many convex decompositions in pure states, a mixed quantum state is not uniquely linked to an ensemble albeit that there is a preferred decomposition, namely the spectral decomposition. A second reason is that the map $\Gamma$ not only mixes pure states but also rotates them which prohibits a description in terms of paths.   

This makes quantifying the randomness of a CP $\Gamma$ not evident. At least two proposals can be found in the literature: the minimal output entropy~\cite{King01} and the map entropy~\cite{Zyczkowski_duality_2004}. Here we propose an approach in the spirit of Markov chains.

\begin{construction}
\label{con1}
Let $\Gamma: \c M_d \to \c M_d$ be a unity preserving CP map with invariant state $\rho$ and let $\Lambda$ be a unity preserving CP map from $\c M_d \otimes \c M_d$ to $\c M_d$ that satisfies the compatibility condition
\begin{equation}
\Lambda(X \otimes \idty) = \Lambda(\idty \otimes X) = \Gamma(X),\enskip X \in \c M_d.
\label{qmc}
\end{equation}
The quantum Markov chain defined by $\Lambda$ and $\rho$ is then the finitely correlated state $\omega$~\cite{Fannes92-1} on the semi-infinite quantum spin chain $\otimes^{\Nl_0} \c M_d$ with marginals $\rho_n$ on the sites $[1, 2, \ldots, n]$ given by
\begin{equation}
\begin{split}
\omega\bigl( X_n \bigr) 
&= \tr \bigl( \rho_n X_n \bigr) \\
&= \tr \bigl( \rho\, \Lambda \circ (\Lambda \otimes \id) \circ \cdots \circ (\Lambda \otimes \id \otimes \cdots \otimes \id) (\idty \otimes X_n) \bigr)
\end{split}
\label{fcs}
\end{equation}
where $X_n \in \lotimes_1^n \c M_d$.\begin{flushright}
$\square$
\end{flushright}
\end{construction}
Note that this construction contains the class of classical hidden Markov processes.

We can now associate an entropy to a quantum Markov chain as in the classical case
\begin{equation}
\s h(\Lambda,\rho) := \s s(\omega) := \lim_{n \to \infty} \frac{1}{n}\, \s S(\rho_n)
\label{qme}
\end{equation}
where $\s S$ is the usual von~Neumann entropy. Generically, $\Gamma$ has a non-degenerate eigenvalue 1 so that $\rho$ is uniquely determined by $\Gamma$ and that there is no $\rho$ dependence in~(\ref{qme}). Clearly a number of issues have to be addressed: for which $\Gamma$ can one find a $\Lambda$ that satisfies~(\ref{qmc})? How does $\s h(\Gamma, \Lambda)$ depend on $\Lambda$? Can one compute the mean entropy~(\ref{qme})?

\section{A Fermionic model}
\label{sec:fermionic}

Quantum states are mostly indirectly given, typically as ground or equilibrium states for a given interaction and are hence difficult to work with as there is for example no explicit density matrix. Also, in general one has to deal with an enormous amount of parameters when the number of particles grows. As the number of components grows, typically the number of parameters grows exponentially. Free Fermionic states~\cite{Shale64,Powers70,Balslev68} form an exception in both respects. These states describe systems of non-interacting fermions. They are given by an explicit recipe, reducing the calculation of higher order correlation to a simple combinatorial combination of second order correlations. Hence not only can they be calculated explicitly, they are also fully described by their second order correlation, resulting in a significant reduction in parameters. 

In this section we first discuss some of the properties of free Fermionic states. We then introduce a Markov construction similar to the one given in Section~\ref{sec:quant}.

\subsection{Free states and maps}
\label{s3:1}

The algebra $\g A(\g H)$ generated by the canonical anticommutation relations (CAR) describes the observables of a system of Fermions with one-particle space $\g H$. It is generated by the identity and the creation and annihilation operators $a^*(\varphi)$ and $a(\varphi)$ that obey the relations
\begin{align}
&\phantom{i}i)\enskip \varphi \mapsto a^*(\varphi) \text{ is $\Cx$-linear} \\
&ii)\enskip \{ a(\varphi) \,,\, a(\psi) \} = 0 \enskip\text{and}\enskip \{ a(\varphi) \,,\, a^*(\psi) \} = \< \varphi \,,\, \psi \> \idty.
\end{align}

A useful set of states, called free, quasi-free, Gaussian, or determinantal, is determined by a simple combinatorial rule. Given a symbol $Q \in \g B(\g H)$, the state $\omega_Q$ vanishes on every monomial except for
\begin{equation}
\omega_Q\bigl( a^*(\varphi_1) \cdots a^*(\varphi_n) a(\psi_n) \cdots a(\psi_1) \bigr) = \det\bigl( \bigl[\< \psi_k \,,\, Q \varphi_\ell \> \bigr] \bigr).
\label{qfs}
\end{equation} 
Positivity holds if and only if $0 \le Q \le \idty$. The set of symbols
\begin{equation}
\c Q(\c H) = \{ Q \mid Q \text{ linear operator on } \c H \text{ such that } 0 \le Q \le \idty\}
\end{equation}
is convex and weakly compact. Convexity at the level of symbols is very different from convexity at the level of the free states. Nevertheless it can be shown that a free state is pure, i.e.\ extreme in the full state space of $\c A(\c H)$, if and only if its symbol is an orthogonal projector, i.e.\ an extreme point of $\c Q(\c H)$.

Important quantities like the entropy of free states are expressible in terms of symbols, e.g.\ 
\begin{equation}
\s S(Q) = - \tr Q \log Q - \tr (\idty - Q) \log (\idty - Q).
\label{qfent}
\end{equation}

Let $P$ be an orthogonal projection on $\c H$, then the restriction of the free state $\omega_Q$ on $\c A(\c H)$ is a free state on the sub-CAR algebra $\c A(P\c H)$ with symbol $PQP$. Conversely, a pair of free states $\omega_{Q_i}$ on $\c A(\c H_i)$, $i=1,2$ extends to a free state $\omega_{Q_1 \oplus Q_2}$ on $\c A(\c H_1 \oplus \c H_2)$ by putting
\begin{equation}
\omega_{Q_1 \oplus Q_2} (X_1 X_2) = \omega_{Q_1}(X_1)\, \omega_{Q_2}(X_2),\enskip X_i \in \c A(\c H_i).
\end{equation} 

Free, identity preserving, CP maps $\Lambda_{A,B}: \g A(\g H) \to \g A(\g K)$ are determined by a pair of linear operators $A: \g H \to \g K$ and $B: \g H \to \g H$. For monomials of degree two
\begin{equation}
\Lambda_{A,B} \bigl( a^*(\varphi)a(\psi) \bigr) = a^*(A\varphi) a(A\psi) + \< \psi \,,\, B\varphi \> \idty.
\label{qfm}
\end{equation} 
For more details, see~\cite{Dierckx08}. Complete positivity holds if and only if $0 \le B \le \idty - A^*A$. As for free states, we introduce the set
\begin{equation}
\begin{split}
\c{CP}(\c H, \c K) = \bigl\{ (A,B) \,\bigm|\, 
&A: \c H \to \c K \text{ and } B: \c H \to \c H \text{ linear} \\
&\text{operators such that } 0 \le B \le \idty - A^*A  \bigr\}.
\end{split}
\end{equation}
We use $\c{CP}(\c H)$ for $\c{CP}(\c H, \c H)$. The set of free, CP maps extends that of free states by putting
\begin{equation}
Q \in \c Q(\c H) \mapsto (0,Q) \in \c{CP}(\c H, \c K).
\end{equation}
Another special distinguished class of maps are the free homomorphism from $\c A(\c H)$ to $\c A(\c K)$
\begin{equation}
\{ (V,0) \in \c{CP}(\c H, \c K) \mid V: \c H \to \c K \text{ isometric} \}.
\end{equation}
The set $\c{CP}(\c H, \c K)$ is also convex and weakly compact. Free CP maps transform free states into free states and one checks from~(\ref{qfs}) and~(\ref{qfm}) that
\begin{equation}
\omega_Q \circ \Lambda_{A,B} = \omega_{A^*QA + B}.
\end{equation}

The construction of the quantum Markov process consists of using a completely positive map to contract the observable and then applying a single-party state that
is invariant under a completely positive map. We have the following lemma concerning the existence of
such invariant states.
\begin{lemma}
\label{lem2}
Let $\Lambda_{A,B}$ be a completely positive free transformation of $\c A(\c H)$ as in~(\ref{qfm}) and assume that $\dim(\c H) < \infty$, then $\Lambda_{A,B}$ has a unique invariant state if and only if $\norm A < 1$. Moreover, the unique invariant state is free with symbol $Q$ determined by
\begin{equation}
Q = A^*Q\,A + B.
\label{lem2:1}
\end{equation}
\end{lemma}

\begin{proof}
The condition $\norm A < 1$ is equivalent to the non-existence of non-trivial solutions to the homogeneous equation $Q = A^*Q\, A$. It has to be satisfied to have uniqueness of the solution of the invariance condition~(\ref{lem2:1}) for symbols. Conversely, suppose that $\norm A < 1$, then there exists by the fixed point theorem for contractions a unique $Q$ such that 
\begin{equation}
Q = A^*Q\,A + B.
\end{equation}
This $Q$ satisfies $0 \le Q \le \idty$ as we may obtain $Q$ by iterating the map $X \mapsto A^*X\,A + B$ with initial value 0. It is then easily checked that
\begin{equation}
\lim_{n\to\infty} \Lambda_{A,B}^n = \omega_Q
\end{equation}
which guarantees both the uniqueness of the invariant state and its free character.
\end{proof}

\subsection{Constructing a chain}
\label{s3:2}

We now have the necessary ingredients to introduce the free Fermionic counterpart of Construction~\ref{con1}. There are natural embeddings
\begin{equation}
a(\varphi) \mapsto a(\varphi \oplus 0) \enskip\text{and}\enskip a(\psi) \mapsto a(0 \oplus \psi)
\label{car:emb}
\end{equation} 
of $\g A(\c H)$ and $\g A(\c K)$ into $\g A(\c H \oplus \c K)$. Both factors together generate the large algebra and they satisfy graded commutation relations as creation operators in different factors anticommute. We can transport the construction of the quantum Markov chain~(\ref{qmc},\ref{fcs}) and its entropy~(\ref{qme}) to the free Fermionic setting. The spin chain algebra $\otimes^{\Nl} \c M_d$ is replaced by a semi-infinite Fermionic chain $\g A(\loplus^{\Nl} \c H)$ where $\g A(\c H)$ is now the one site algebra. 

The basic ingredient is a free CP transformation $\Lambda_{A,B}$ of $\g A(\c H)$ and we look for free CP maps $\Lambda_{C,D}$ from $\g A(\c H \oplus \c H)$ to $\g A(\c H)$ such that
\begin{equation}
\label{eq:qf_compatibility}
\Lambda_{C,D} \circ \jmath_1 = \Lambda_{C,D} \circ \jmath_2 = \Lambda_{A,B}.
\end{equation}
Here, $\jmath_1$ and $\jmath_2$ are the natural embeddings of $\g A(\c H)$ into the first and second factor of $\g A(\c H \oplus \c H)$
\begin{equation}
\jmath_1(a(\varphi)) = a(\varphi \oplus 0) \enskip\text{and}\enskip \jmath_2(a(\varphi)) = a(0 \oplus \varphi).
\end{equation}
Applying the compatibility condition~(\ref{eq:qf_compatibility}) to monomials $a(\varphi)$ and $a^*(\varphi)a(\psi)$ we see that
\begin{equation}
C = \begin{bmatrix} A & A \end{bmatrix} \enskip\text{and}\enskip D = \begin{bmatrix} B & X \\ X^* & B \end{bmatrix},
\label{qfcomp}
\end{equation}
where $X$ is as of yet undetermined and allows for some freedom in the choice of $D$. Because of the structure of free CP maps, the compatibility conditions~(\ref{qfcomp}) are not only necessary but also sufficient and we can rephrase the whole construction on the level of symbols. Doing so, graded tensor products become direct sums.

\begin{construction}
\label{con3}
Let $(A,B) \in \c{CP}(\c H)$ and let $Q \in \c Q(\c H)$ be such that $\omega_Q$ is invariant under $\Lambda_{A,B}$:
\begin{equation}
Q = A^* Q A + B.
\label{inv}
\end{equation}
Let $X: \c H \to \c H$ satisfy the compatibility condition
\begin{equation}
(C,D) \in \c{CP}(\c H \oplus \c H, \c H) \enskip\text{with}\enskip C \text{ and } D \text{ as in~(\ref{qfcomp})}.
\label{qfcomp2}
\end{equation}
The free Markov chain defined by $X$ and $Q$ is the symbol
\begin{equation}
Q_\infty = \underset{n\to\infty}{\text{w-lim}}\, P_n R_n P_n^*\enskip \text{on } \oplus^{\Nl} \c H 
\label{qinf}
\end{equation}
where
\begin{align}
&P_n: \c H \oplus \Bigl( \oplus_{k=0}^{n-1} \c H \Bigr) \to \Bigl( \oplus_{k=0}^{n-1} \c H \Bigr): \varphi \oplus \psi_n \mapsto \psi_n \\ 
&R_0 = Q \enskip\text{and}\enskip R_{n+1} = \Bigl( C^* \oplus \left( \oplus^n \idty \right) \Bigr) R_n \Bigl(C \oplus \left( \oplus^n \idty \right) \Bigr) + \Bigl( D\oplus \left( \oplus^n 0 \right) \Bigr).
\label{rn}
\end{align}
\begin{flushright}
$\square$
\end{flushright}
\end{construction}  


There is some freedom in choosing the channel $\Lambda_{C,D}$.
The operator $X$ has to be chosen such that the $\Lambda_{C,D}$ is completely positive, i.e.
$0 \leq D \leq \idty - C^* C$. One may wonder if and when this is possible. The question of existence of compatible channels is answered by the following lemma.
\begin{lemma}
\label{lem1}
The compatibility condition~(\ref{qfcomp2}) is satisfiable if and only if
\begin{equation}
A^*A \le \min \bigl( \{\tfrac{1}{2}\, \idty, \idty - B \}\bigr).
\end{equation}
\end{lemma}

\begin{proof}
We look for the necessary and sufficient conditions for the existence of a $X: \c H \to \c H$ such that
\begin{equation}
\begin{bmatrix} B &X \\ X^* &B \end{bmatrix} \ge 0 \enskip\text{and}\enskip \begin{bmatrix} \idty - A^*A - B &-A^*A - X \\ -A^*A - X^* &\idty - A^*A - B \end{bmatrix} \ge 0.
\end{equation}
Clearly $0 \le B \le \idty - A^*A$ as $(A,B) \in \c{CP}(\c H)$. The remaining positivity conditions are then the existence of $S$ and $T$ with 
\begin{equation}
\begin{split}
&\norm S \le 1,\enskip \norm T \le 1, \enskip X = B^{\frac{1}{2}} S B^{\frac{1}{2}}, \enskip \text{and } \\
&A^*A + X = (\idty - A^*A - B)^{\frac{1}{2}} T (\idty - A^*A - B)^{\frac{1}{2}}.
\end{split}
\end{equation}
Replacing $S$ and $T$ by their Hermitian parts, we may restrict to Hermitian $X$ and so we need 
\begin{equation}
[-B \,,\, B] \cap [-\idty + B \,,\, \idty - 2A^*A - B ] \neq \emptyset
\end{equation} 
or, equivalently, that
\begin{equation}
[\idty \,,\, \idty + 2B] \cap [2B \,,\, 2 \idty - 2A^*A ] \neq \emptyset.
\end{equation} 
But this is the case if and only if
\begin{equation}
\max \bigl( \{ \idty, 2B \}\bigr) \le 2\idty - 2A^*A \enskip\text{or}\enskip A^*A \le \min \bigl( \{\tfrac{1}{2}\, \idty, \idty - B \}\bigr).
\end{equation}
\end{proof}

Let us look at this compatibility condition for a simple case.
\begin{example}
\label{ex:compat}
If $\c H = \Cx$, we have
\[D = \begin{bmatrix} b & x \\ \overline{x} & b \end{bmatrix} \text{ and } C = \begin{bmatrix} a & a  \end{bmatrix} \,, \]
with $a,x \in \Cx$ and $b \in \Rl$. From the complete positivity of $\Lambda_{a,b}$, we know that $0 \leq b \leq 1-|a|^2$. The complete positivity conditions for $C$ and $D$ limit the choice for $x$. From $0 \leq D$ we see that
\[|x| \leq b \,.\]
From $D \leq \idty - C^* C $ on the other hand, we get that
\[ | x + |a|^2| \leq 1 - |a|^2 -b \,. \]
These two inequalities means that $x$ has to lie in the intersection of two circles in the complex plane, one centred at $0$ with radius $r_1=b$ and another one centred at $-|a|^2$ with radius $r_2=1-|a|^2 - b$. These two circles have an intersection when the distance between the centres is smaller than the sum of the radii. Hence, the channel is extendible if
\[ |a|^2 \leq \frac{1}{2} \]
which corresponds to the conditions in the lemma.
\end{example}

Given the constituents of the Markov construction, the channel $\Lambda_{C,D}$ and the invariant symbol $Q$, the symbol of the full process $Q_\infty$ can easily be determined
\begin{proposition}
\label{pro1}
The symbol $Q_\infty$ in~(\ref{qinf}) is an Hermitian block Toeplitz matrix with entries
\begin{equation}
\bigl( Q_\infty \bigr)_{i\,i} = Q \enskip\text{and}\enskip \bigl( Q_\infty \bigr)_{i\,i+n} = (A^*)^n (Q - B + X).
\end{equation}
Here $i=0,1,2,\ldots$ and $n=1,2,3,\ldots$
\end{proposition}

\begin{proof}
The proof consists in a straightforward computation of the consecutive $R_n$ in~(\ref{rn}) combined with the invariance equation~(\ref{inv}).
\end{proof}

\section{Fermionic Entropy density}
\label{sec:density}

In this section we compute the entropy $h$ for the Fermionic Markov process constructed in Section~\ref{sec:fermionic}. We can associate an entropy to a Fermionic Markov chain using~(\ref{qfent})
\begin{align}
\s h(X,Q) 
&:= \lim_{n \to \infty} \frac{1}{n}\, \s S(P_n R_n P_n^*) \\
&:= \lim_{n \to \infty} \frac{1}{n}\, \Bigl( - \tr P_n R_n P_n^* \log(P_n R_n P_n^*) \nonumber \\
&\phantom{:= \lim_{n \to \infty} \frac{1}{n}\, \Bigl(\ }- \tr (\idty - P_n R_n P_n^*) \log(\idty - P_n R_n P_n^*) \Bigr).
\label{qfme}
\end{align}

A first method to compute this relies directly on the expression~(\ref{qfent}) for the entropy of a free state in terms of its symbol and on the structure of the symbols $Q_\infty$ in Proposition~\ref{pro1}. A second way is to rewrite the entropy as the asymptotic rate of disorder, as in the classical case, see Section~\ref{sec:class}. This last approach was used in~\cite{Blackwell57,Fannes92} to compute the entropy of a hidden Markov process. The first method uses the full local restrictions of the state while the second relies on the incremental structure of the local states given by a transfer matrix like construction, see~(\ref{fcs}) and~(\ref{rn}). 

\subsection{Direct approach}
The first approach to calculating the entropy density uses an extension of Szeg\"o's theorem to block Toeplitz matrices $\hat T$. This theorem allows to calculate asymptotic densities of trace functions of Toeplitz matrices. A block Toeplitz matrix is a block matrix $\hat T$ where the blocks along diagonals are equal
\[ {\hat T}_{i,j} = {\hat T}_{i+k,j+k} \,, \]
where ${\hat T}_{i,j}$ denotes a block elements. Using Szeg\"o's theorem, we can write densities
\[ \lim_{n \rightarrow \infty} \tr \frac{f(\hat{T}_n)}{n} \]
of a matrix function $f$ and the finite projections ${\hat T}_n = P_n {\hat T} P_n$ in terms of a generating function $T(\theta)$. The Fourier coefficients of this $T(\theta)$ are the elements on the diagonals of $\hat T$. We will now formulate this more precisely.

Let $T: [-\pi,\pi[ \to \c M_d$ be an essentially bounded measurable matrix-valued function on the circle and denote its Fourier coefficients by
\begin{equation*}
\hat T(k) := \frac{1}{2\pi}\, \int_{-\pi}^\pi \!d\theta\, T(\theta)\, \r e^{-ik\theta} \in \c M_d.
\end{equation*}
A function $T$ is essentially bounded if there exist a constant $M$ such that $|T(\theta)| \leq M$ almost everywhere. 
The operator
\begin{equation*}
\hat T = \begin{pmatrix}
\hat T(0)     &\hat T(1)    &\hat T(2) &\ldots \\
\hat T(-1)  &\hat T(0)    &\hat T(1) &\ldots \\
\hat T(-2)  &\hat T(-1) &\hat T(0) &\ldots \\
\vdots  & \vdots & \vdots &\ddots    \\
\end{pmatrix} 
\end{equation*}
defined on $\ell^0(\Nl) \otimes \Cx^d$ extends to a bounded linear transformation of $\ell^2(\Nl) \otimes \Cx^d$. Operators of this type are block Toeplitz matrices and one has
\[
\norm{\hat T} = \norm T_\infty = \underset{\theta}{\text{ess\,sup}}\, \norm{T(\theta)} \,,
\]
where the essential supremum of $T$ is the infimum of all constants $M$ that bound $|T(\theta)|$ almost everywhere.

The Toeplitz matrices we are interested in are symbols and hence self-adjoint. For such Toeplitz matrices, we have that ${\hat T}^* = \hat T$ if and only if the function $T$ takes values in the Hermitian matrices.

\subsubsection{Szeg\"o's theorem}
An extension of Szeg\"o's theorem to block Toeplitz matrices characterizes the limiting spectrum of principal submatrices $P_n \hat T\, P_n$ in terms of the generating function $T$, see~\cite{Miranda00}. Here $P_n$ projects on the first $n$ blocks in $\ell^2(\Nl) \otimes \Cx^d$. We obtain here a more general characterization of such limiting submatrices.

Let us denote for a simply connected compact subset $\c K$ of $\Cx$ by $\c H(\c K)$ the set of continuous functions $f: \c K \to \Cx$ that are holomorphic in the interior $\overset{\circ}{\c K}$ of $\c K$.
Mergelyan's theorem~\cite{Rudin87} asserts that the complex polynomials in the indeterminate $z$ are dense in $\c H(\c K)$: for any $f \in \c H(\c K)$ and $\epsilon > 0$ there exists a polynomial $p^\epsilon$ such that
\begin{equation*}
\max_{z \in \c K} \bigl| f(z) - p^\epsilon(z) \bigr| \le \epsilon.
\end{equation*}
Finally, let us denote by $\Ex_n$ the conditional expectation from $\c B \bigl( \ell^2(\Nl) \bigr) \otimes \c M_d \to \c M_d$ which traces out the first $n$ blocks
\begin{equation*}
\Ex_n(X) := \frac{1}{n}\, \sum_{j=0}^{n-1} X_{jj} \in \c M_d.
\end{equation*} 

We get the following generalization of Szeg\"o's theorem~\cite{Grenander01}.
\begin{thm}
\label{thm:szego}
Let $\{T_1, T_2, \ldots, T_k \} \subset \c L^\infty \bigl(  [-\pi, \pi[, \c M_d \bigr)$ be such that every $T_j(\theta)$ is $\theta$-a.e.\ diagonalizable, let $f_j \in \c H \bigl( \{ z \in \Cx \mid \abs z \le \norm{T_j}_\infty \} \bigr)$ for $j = 1, 2, \ldots, k$ and let $A_j \in \c M_d$, $j = 1, 2, \ldots, k+1$, then
\begin{equation}
\begin{split}
&\lim_{n \to \infty} \Ex_n \Bigl( (\idty \otimes A_1)\, f_1 \bigl( P_n \hat T_1 P_n \bigr)\, (\idty \otimes A_2)\, \cdots f_k \bigl( P_n \hat T_k P_n \bigr)\, (\idty \otimes A_{k+1}) \Bigr) \\
&\quad= \frac{1}{2\pi}\, \int_{-\pi}^\pi \!d\theta\, A_1\, f_1(T_1(\theta))\, A_2 \cdots f_k(T_k(\theta))\, A_{k+1}.
\end{split}
\label{thm:szego:1}
\end{equation}
\end{thm}

\begin{proof}
The proof relies on a continuity argument combined with a standard counting argument. First remark that given $\epsilon > 0$ every $f_j$ can be approximated by a suitable complex polynomial $p^\epsilon_j$
\begin{equation*}
\max_{\abs z \le \norm{T_j}_\infty} \bigl| f_j(z) - p^\epsilon_j(z) \bigr| \le \epsilon.
\end{equation*}
Next, as $\norm{T_j(\theta)} \le \norm{T_j}_\infty$ a.e., we can use von~Neumann's inequality~\cite{Halmos74} to get
\begin{align}
&\norm{f_j(T_j(\theta))} \le \max_{\abs z \le \norm{T_j}_\infty} \abs{f_j(z)} \enskip\text{and} \\
&\bigl\Vert f_j(T_j(\theta)) - p^\epsilon_j(T_j(\theta)) \bigr\Vert = \bigl\Vert \bigl( f_j - p^\epsilon_j \bigr)(T_j(\theta)) \bigr\Vert 
\nonumber \\
&\quad \le \max_{\abs z \le \norm{T_j(\theta)}} \bigl( f_j - p^\epsilon_j \bigr)(z) \le \max_{\abs z \le \norm{T_j}_\infty} \bigl( f_j - p^\epsilon_j \bigr)(z) \le \epsilon.
\end{align}
These estimates allow to replace the $f_j$ in~(\ref{thm:szego:1}) by polynomials. It then remains to verify the statement for monomials, but this reduces to a standard counting argument.

In the case where there is only one function $f(X)=X^k$ and $A_j = \idty$, the density limit can be worked out as follows:
\begin{align}
\lim_{n \rightarrow \infty} \frac{1}{n} \tr (P_n  \hat T P_n)^k &= \lim_{n \rightarrow \infty} \frac{1}{n} \sum_{i_1,\ldots, i_k =0}^n \hat T_{i_1, i_2} \hat T_{i_2, i_3} \ldots \hat T_{i_k, i_1} \\
&= \lim_{n \rightarrow \infty} \frac{1}{n} \sum_{i_1,\ldots, i_k =0}^n \hat T( i_2 - i_1 ) \hat T( i_3 - i_2 ) \ldots \hat T( i_k - i_1 )
\end{align}
By substituting $v_1 = i_2 - i_1 \,, \ldots \,, v_{k-1} = i_k - i_{k-1}$, this sum becomes:
\[
\lim_{n \rightarrow \infty} \frac{1}{n} \sum_{v_1, \ldots, v_{k-1} = -n}^n \,\, \sum_{i_1 \in \mathcal{S}_n(v_1,\ldots, v_{k-1})} \hat T(v_1) \ldots \hat T(v_{k-1}) \hat T(-v_1 - \ldots - v_{k-1}) \,,
\]
where $\mathcal{S}_n(v_1,\ldots, v_{k-1})$ is the set of indices $i$ such that $v_1+i, v_1+v_2+i, \ldots, v_1 + \ldots v_{k-1} + i \in [0,n]$. The number of elements in this set increases by exactly one when $n$ goes to $n+1$, so in the limit we get
\[ \lim_{n \rightarrow \infty} \sum_{v_1, \ldots, v_{k-1} = -\infty}^\infty \hat T(v_1) \ldots \hat T(v_{k-1}) \hat T(-v_1 - \ldots - v_{k-1}) \,. \]
This is exactly the zeroth Fourier coefficient of $T(\theta)^k$, so we get that the density equals
\[ \frac{1}{2 \pi} \int_{-pi}^\pi d \theta f(T(\theta)) \,. \]
The general case of the theorem can be worked out in a similar manner.
\end{proof}

To deal with entropy we don't need the full amalgamated extension of Theorem~\ref{thm:szego} of Szeg\"o's theorem but we may restrict ourselves for an Hermitian $T$ to the asymptotic eigenvalue distribution of the principal blocks $P_n \hat T P_n$. Taking the trace of~(\ref{thm:szego:1}) with a single $f$ and all $A_j = \idty$ we recover the result~\cite{Miranda00}. We denote by $\inf(T)$ and $\sup(T)$ the largest and smallest real numbers such that
\begin{equation}
\inf(T) \le T \le \sup(T)\enskip \text{a.e.}
\end{equation}
The increasingly ordered eigenvalues $\bigl( \tau_1(\theta), \tau_2(\theta), \ldots, \tau_d(\theta) \bigr)$ of $T(\theta)$ are measurable functions of $\theta$ that satisfy
\begin{equation}
\inf(T) \le \tau_1(\theta) \le \cdots \le \tau_d(\theta) \le \sup(T).
\end{equation}
The eigenvalue distribution of $P_n \hat T P_n$ is the atomic probability measure
\begin{equation}
\delta_n = \frac{1}{nd}\, \sum_{\lambda \in \sigma(P_n\hat T P_n)} \delta_\lambda.
\end{equation}

\begin{thm}
\label{thm1}
With the assumptions of above
\begin{equation} 
\underset{n \to \infty}{\r{w^*\!\text{-}\,lim}}\, \delta_n = \delta_\infty,
\end{equation} 
where
\begin{equation}
\delta_\infty \bigl( ]-\infty, t] \bigr) = \frac{1}{d}\, \sum_{k=1}^d \frac{1}{2\pi}\, \int_{\tau_k(\theta) \le t} \!d\theta.
\label{limdis}
\end{equation}
\end{thm}

An equivalent way to express this result is saying that for any continuous complex function $f$ on $[\inf(T), \sup(T)]$
\begin{equation}
\lim_{n \to \infty} \frac{1}{nd}\, \tr f(P_n \hat T P_n) = \frac{1}{2\pi} \int_{-\pi}^\pi \!d\theta\, \frac{1}{d}\, \tr f(T(\theta)). 
\label{szego2}
\end{equation}
This version is in some sense more natural as it doesn't involve the reordering of the eigenvalue functions $\tau_k$ used in the definition of the distribution function of the limiting eigenvalue distribution~(\ref{limdis}).

We can apply Theorem~\ref{thm1} to the computation of the entropy, replacing the Toeplitz operator $T$ by $Q_\infty$ in Proposition~\ref{pro1} and choosing in~(\ref{szego2})
\begin{equation}
f(\lambda) = -\lambda \log(\lambda) - (1-\lambda)\log(1-\lambda)\enskip \text{on } (0,1).
\end{equation}
The generating function $T$ becomes
\begin{align*}
\theta \mapsto Q + &(Q -B +X) A \r e^{\imath \theta}(\idty - A \r e^{\imath \theta})^{-1} \\ &+ A^* \r e^{-\imath \theta}(\idty - A^* \r e^{-\imath \theta})^{-1} (Q -B +X^*).
\end{align*}

\begin{example}
In case of a single particle space $\mathcal{H}=\mathbb{C}$ as in the Example~\ref{ex:compat}, the entropy can be calculated from the scalar version of the above function
\[ \theta \mapsto q + (q -b +x) a \r e^{\imath \theta}(1 - a \r e^{\imath \theta})^{-1} \\ + \text{h.c.} \, \]
where $x$ lies within the two circles determining the compatibility condition, as explained in Example~\ref{ex:compat}.

This scalar function is linear in $x$ and the function $f$ is concave in it's argument. Hence the minimal entropy is obtained on the border of the compatibility region, much like in the case of the classical Markov process and compatible hidden Markov processes described at the end of Section~\ref{sec:class}.
\end{example}

\subsection{Entropy rate approach}
The second approach expresses the entropy as an asymptotic rate. Let $\omega$ be a translation invariant state on a quantum spin chain $\otimes^\Ir \c M_d$ and denote by $\rho_{(0,n-1)}$ its reduced density matrices, i.e.\
\begin{equation}
\omega(X) = \tr \bigl( \rho_{(0,n-1)} X \bigr)\enskip \text{for } X \in \otimes_{k=0}^{n-1} \c M_d.
\end{equation}
As we have seen before, subadditivity combined with translation invariance guarantee the existence of the mean entropy of $\omega$ for intervals
\begin{equation}
\s s(\omega) = \lim_{n \to \infty} \frac{1}{n}\, \s S(\rho_{(0,n-1)}).
\end{equation}
Moreover, strong subadditivity in conjunction with translation invariance also guarantees that 
\begin{align}
&n \mapsto \s S(\rho_{(0,n-1)})\enskip \text{is monotonically increasing and} 
\label{ent:mon} \\
&\s s(\omega) = \lim_{n \to \infty} \frac{1}{n}\, \s S(\rho_{(0,n-1)}) = \lim_{n \to \infty} \Bigl( \s S(\rho_{(0,n)}) - \s S(\rho_{(0,n-1)}) \Bigr).
\label{ent:hop}
\end{align}
Both properties~(\ref{ent:mon}) and~(\ref{ent:hop}) fail for general quantum states or for general finite local regions~\cite{Kay01}. These results for quantum spin chains extend to Fermionic lattices using the natural embeddings~(\ref{car:emb}) and restricting to even states~\cite{Araki03}. The equality of both limits in~(\ref{ent:hop}) can be seen as a discrete version of de~l'H\^opital's rule. Obviously, the existence of the limit of the differences is a much stronger requirement than that of the averages. 

For free Fermionic states we can work at the level of symbols. E.g., strong subadditivity of entropy amounts to 
\begin{equation}
\s S(Q_{123}) + \s S(Q_2) \le \s S(Q_{12}) + \s S(Q_{23})
\end{equation}  
where $\s S$ is defined in~(\ref{qfent}) and where the symbols in the inequality are as follows
\begin{equation}
 Q_{123} = \begin{bmatrix} Q_1 &T &S \\ T^* &Q_2 &R \\ S^* &R^* &Q_3 \end{bmatrix},\enskip Q_{12} = \begin{bmatrix} Q_1 &T \\ T^* &Q_2 \end{bmatrix}, \enskip\text{and}\enskip Q_{23} = \begin{bmatrix} Q_2 &R \\ R^* &Q_3 \end{bmatrix}.
\end{equation}
For more on functions that satisfy such strong subadditivity, see \cite{audenaert_strongly_2010}.

Below, we extend the equality of the limit of differences with that of averages, as in~(\ref{ent:hop}), to a much wider class of functions than the strongly subadditive ones, like the entropy of a symbol~(\ref{qfent}). The argument relies on regularity of the functions and not on subadditivity or convexity which rarely hold. Szeg\"o's theorem follows as a consequence. 

We first show that the theorem holds for polynomials.
\begin{lemma}
\label{lem3}
With the notation and assumptions on an Hermitian Toeplitz operator at the beginning of this section, for any polynomial $p$ 
\begin{equation} 
\lim_{n \to \infty} \Bigl( \tr p(P_{n} \hat T P_{n}) - \tr p(P_{n-1} \hat T P_{n-1}) \Bigr) = \frac{1}{2\pi} \int_{-\pi}^\pi \!d\theta\, \tr p(T(\theta)).
\end{equation} 
\end{lemma}

\begin{proof}
It suffices to consider $p(\lambda) = \lambda^k$ for $k \in \Nl$. We have
\begin{align*}
&\lim_{n \to \infty} \tr (P_n \hat T P_n)^k - \tr (P_{n-1} \hat T P_{n-1})^k \\
&\quad= \lim_{n \to \infty} \sum_{i=1}^n \tr \bigl( (P_n \hat T P_n)^k \bigr)_{ii} - \sum_{i=1}^{n-1} \tr \bigl( (P_{n-1} \hat T P_{n-1})^k) \bigr)_{ii} \\
&\quad= \lim_{n \to \infty} \Bigl( \sum_{i_1,\ldots,i_k=1}^n - \sum_{i_1,\ldots,i_k=1}^{n-1} \Bigr) \\
&\qquad \qquad \qquad \tr \Bigl\{ (P_n \hat T P_n)_{i_1i_2}  \cdots (P_n \hat T P_n)_{i_{k-1}i_k} (P_n \hat T P_n)_{i_ki_1} \Bigr\} \\
&\quad= \lim_{n \to \infty} \Bigl( \sum_{i_1,\ldots,i_k=1}^n - \sum_{i_1,\ldots,i_k=1}^{n-1} \Bigr) \hat T (i_2 - i_1) \ldots \hat T (i_1 - i_k) \,,
\end{align*}
where $(Q)_{ij}$ denotes the block at position $(i,j)$ inside of a block matrix $Q$.

By substituting $v_1 = i_2 - i_1 \,, \ldots \,, v_{k-1} = i_k - i_{k-1}$, this sum becomes:
\begin{align*}
&\lim_{n \rightarrow \infty} \sum_{v_1, \ldots, v_{k-1} = -n}^n \,\, \Bigl(\sum_{i_1 \in \mathcal{S}_n(v_1,\ldots, v_{k-1})} - \sum_{i_1 \in \mathcal{S}_{n-1}(v_1,\ldots, v_{k-1})} \Bigr) \\ 
& \qquad \qquad \qquad \qquad \qquad \qquad \hat T(v_1) \ldots \hat T(v_{k-1}) \hat T(-v_1 - \ldots - v_{k-1}) \,,
\end{align*}
where $\mathcal{S}_n(v_1,\ldots, v_{k-1})$ is the set of indices $i$ such that $v_1+i, v_1+v_2+i, \ldots, v_1 + \ldots v_{k-1} + i \in [0,n]$. For fixed $v_1, \ldots, v_{k-1}$, the number of elements in these sets increases by exactly one when $n$ goes to $n+1$. Hence, the difference of sums between brackets equals one and we arrive at the expression prescribed by the lemma.
\end{proof}

We can now use this lemma and an approximation argument to prove the general case.
\begin{thm}
\label{thm:lhopital}
With the notation and assumptions on an Hermitian Toeplitz operator at the beginning of this section, for any function $f$ that is absolutely continuous on the interval $[ \inf(T), \sup(T) ]$
\begin{equation} 
\lim_{n \to \infty} \Bigl( \tr f(P_{n+1} \hat T P_{n+1}) - \tr f(P_n \hat T P_n) \Bigr) = \frac{1}{2\pi} \int_{-\pi}^\pi \!d\theta\, \tr f(T(\theta)).
\end{equation} 
\end{thm}

\begin{proof}
By the continuity of the eigenvalues of a matrix and by the minimax principle~\cite{Courant89} we can label the eigenvalues of $P_n \hat T P_n$ as
\begin{align}
&\{\tau^n_{k\,j} \mid k = 1,2,\ldots, d,\ j = 1,2, \ldots, n\} \enskip \text{with}
\nonumber \\
&\inf(T) \le \tau^n_{1\,j} \le \tau^n_{2\,j} \le \cdots \le \tau^n_{d\,j} \le \sup(T) \enskip\text{and}\enskip \tau^{n+1}_{k\,j} \le \tau^n_{k\,j} \le \tau^{n+1}_{k\,j+1}.
\label{thm:lhopital:0}
\end{align}
See \cite{Horn90} for a proof of this interlacing property.

Let $f: [\inf(T), \sup(T)] \to \Cx$ be absolutely continuous with integrable derivative $g$, then for any $\lambda,\tau \in [\inf(T), \sup(T)]$
\begin{equation}
f(\lambda) = f(\tau) + \int_\tau^\lambda \!dx\, g(x).
\end{equation}
Therefore
\begin{align}
\frac{1}{2\pi}\, \int_{-\pi}^\pi \!d\theta\, f(\tau(\theta))
&= \frac{1}{2\pi}\, \int_{-\pi}^\pi \!d\theta\, \Bigl\{ f(\tau) + \int_\tau^{\tau(\theta)} \! dx\, g(x) \Bigr\} \\
&= f(\tau) + \frac{1}{2\pi}\, \int_{-\pi}^\pi \!d\theta \int_\tau^{\tau(\theta)} \!dx\, g(x) \\
&= f(\tau) + \int_{\inf(T)}^{\sup(T)} \!dx\, g(x)\, \frac{1}{2\pi}\, \int_{-\pi}^\pi \!d\theta\, \eta(\tau,x,\theta).
\label{thm:lhopital:1}
\end{align}
Here, $\eta$ is defined as
\begin{equation}
\eta(\tau,x,\theta) = 
\begin{cases} 1 &\tau < x < \tau(\theta) \\-1 &\tau(\theta) < x < \tau \\0 &\text{otherwise} \end{cases}.
\end{equation}
By~(\ref{thm:lhopital:1}) we rewrite the increment of traces of $f(P_n \hat T P_n)$ as
\begin{align}
&\tr f(P_{n+1} \hat T P_{n+1}) - \tr f(P_n \hat T P_n) \\
&\quad= \sum_{k=1}^d \Bigl\{ \sum_{j=1}^{n+1} f(\tau^{n+1}_{k\,j}) - \sum_{j=1}^n f(\tau^n_{k\,j}) \Bigr\} \\
&\quad= \frac{1}{2\pi}\, \int_{-\pi}^\pi \!d\theta\, \sum_{k=1}^d f(\tau_k(\theta)) - \sum_{k=1}^d \int_{\inf(T)}^{\sup(T)} \!dx\, g(x)\, \frac{1}{2\pi} \int_{-\pi}^\pi \!d\theta 
\nonumber\\ 
&\phantom{\quad= \frac{1}{2\pi}\, \int_{-\pi}^\pi \!d\theta\ } \Bigl\{ \sum_{j=1}^{n+1} \eta(\tau^{n+1}_{k\,j}, x, \tau_k(\theta)) - \sum_{j=1}^n \eta(\tau^n_{k\,j}, x, \tau_k(\theta)) \Bigr\} \\
&\quad= \frac{1}{2\pi}\, \int_{-\pi}^\pi \!d\theta\, \tr f(\hat T(\theta)) - \sum_{k=1}^d \int_{\inf(T)}^{\sup(T)} \!dx\, g(x)\, \frac{1}{2\pi} \int_{-\pi}^\pi \!d\theta\, h^n_k(x,\theta), 
\end{align}
with
\begin{equation}
h^n_k(x,\theta) = \sum_{j=1}^{n+1} \eta(\tau^{n+1}_{k\,j}, x, \tau_k(\theta)) - \sum_{j=1}^n \eta(\tau^n_{k\,j}, x, \tau_k(\theta)).
\end{equation}
The functions $h^n_k$ are piecewise constant with values -1, 0 or 1 due to the interlacement~(\ref{thm:lhopital:0}) of the $\tau^n_{k\,j}$. As any integrable $g$ on $[\inf(T), \sup(T)]$ can be arbitrarily well approximated in $\c L^1$-norm by polynomials, the theorem follows from Lemma~\ref{lem3}.
\end{proof}

\section{Conclusion}
We have studied a free Fermionic version of quantum Markov processes. Due to the free Fermionic nature of the states we can characterize all possible Markov processes that one can construct. The density matrices of these states can be described by a Toeplitz matrix.
By studying the behaviour of the eigenvalues of subsequent Toeplitz matrices, we have proved a new Szeg\"o theorem that allows to calculate the asymptotic entropy rate. This is what corresponds in the free Fermionic case to the method proposed by Blackwell~\cite{Blackwell57}.

It would be interesting to look for other quantum Markov processes for which an
explicit calculation of the entropy rate is possible. Processes with a high symmetry
are obvious first choices. Hopefully, such a calculation can lead to a quantum version of the Blackwell dynamical system.

\bibliographystyle{plain}
\bibliography{bibliography}

\begin{thebibliography}{10}

\bibitem{Accardi81}
L.~Accardi.
\newblock Topics in quantum probability.
\newblock {\em Physics Reports}, 77(3):169--192, November 1981.

\bibitem{Alicki01}
R.~Alicki and M.~Fannes.
\newblock {\em Quantum dynamical systems}.
\newblock Oxford University Press, Oxford, 2001.

\bibitem{Araki03}
H.~Araki and H.~Moriya.
\newblock Equilibrium statistical mechanics of {F}ermion lattice systems.
\newblock {\em Reviews in Mathematical Physics}, 15:93--198, 2003.

\bibitem{audenaert_strongly_2010}
K.~Audenaert, F.~Hiai, and D.~Petz.
\newblock Strongly subadditive functions.
\newblock {\em Acta Mathematica Hungarica}, 128(4):386--394, 2010.

\bibitem{Balslev68}
E.~Balslev, J.~Manuceau, and A.~Verbeure.
\newblock Representations of anticommutation relations and {B}ogolioubov
  transformations.
\newblock {\em Communications in Mathematical Physics}, 8(4):315--326, 1968.

\bibitem{Blackwell57}
D.~Blackwell.
\newblock The entropy of functions of finite state {M}arkov chains.
\newblock In {\em Trans. First Prague Conference on Information Theory,
  Decision Functions, and Random Processes, Prague}, pages 13--20, 1957.

\bibitem{Courant89}
R.~Courant and D.~Hilbert.
\newblock {\em Methods of Mathematical Physics}.
\newblock Wiley-VCH, first edition, 1989.

\bibitem{Dierckx08}
B.~Dierckx, M.~Fannes, and M.~Pogorzelska.
\newblock Fermionic quasifree states and maps in information theory.
\newblock {\em Journal of Mathematical Physics}, 49(3):032109, 2008.

\bibitem{Fannes92}
M.~Fannes, B.~Nachtergaele, and L.~Slegers.
\newblock Functions of {M}arkov processes and algebraic measures.
\newblock {\em Reviews in Mathematical Physics}, 4(1):39, 1992.

\bibitem{Fannes89}
M.~Fannes, B.~Nachtergaele, and R.F. Werner.
\newblock Exact antiferromagnetic ground states of quantum spin chains.
\newblock {\em Europhysics Letters {(EPL)}}, 10(7):633--637, December 1989.

\bibitem{Fannes92-1}
M.~Fannes, B.~Nachtergaele, and R.F. Werner.
\newblock Finitely correlated states on quantum spin chains.
\newblock {\em Communications in Mathematical Physics}, 144(3):443, 1992.

\bibitem{Grenander01}
U.~Grenander and G.~Szeg{\H{o}}.
\newblock {\em Toeplitz Forms and Their Applications}.
\newblock AMS Bookstore, 2001.

\bibitem{Halmos74}
P.R. Halmos.
\newblock {\em A Hilbert space problem book}.
\newblock Springer-Verlag, New York, 1974.

\bibitem{Horn90}
R.A. Horn and C.R. Johnson.
\newblock {\em Matrix Analysis}.
\newblock Cambridge University Press, 1990.

\bibitem{Kay01}
A.R. Kay and B.S. Kay.
\newblock Monotonicity with volume of entropy and of mean entropy for
  translationally invariant systems as consequences of strong subadditivity.
\newblock {\em Journal of Physics A}, 34(3):365, 2001.

\bibitem{King01}
C.~King and M.B. Ruskai.
\newblock Minimal entropy of states emerging from noisy quantum channels.
\newblock {\em Information Theory, IEEE Transactions on}, 47(1):192--209, 2001.

\bibitem{Klumper91}
A.~Kl\"umper, A.~Schadschneider, and J.~Zittartz.
\newblock Equivalence and solution of anisotropic spin-1 models and generalized
  {t-J} fermion models in one dimension.
\newblock {\em Journal of Physics A: Mathematical and General},
  24(16):L955--L959, August 1991.

\bibitem{Miranda00}
M.~Miranda and P.~Tilli.
\newblock Asymptotic spectra of hermitian block {T}oeplitz matrices and
  preconditioning results.
\newblock {\em SIAM J. Matrix Anal. Appl.}, 21(3):867--881, 2000.

\bibitem{Perez07}
D.~Perez-Garcia, F.~Verstraete, M.~M. Wolf, and J.~I. Cirac.
\newblock Matrix product state representations.
\newblock {\em Quantum Info. Comput.}, 7(5):401--430, July 2007.

\bibitem{Powers70}
R.T. Powers and E.~St{\o}rmer.
\newblock Free states of the canonical anticommutation relations.
\newblock {\em Communications in Mathematical Physics}, 16(1):1--33, 1970.

\bibitem{Rudin87}
W.~Rudin.
\newblock {\em Real and complex analysis}.
\newblock McGraw-Hill Book Co., New York, third edition, 1987.

\bibitem{Shale64}
D.~Shale and W.~F. Stinespring.
\newblock States of the {C}lifford algebra.
\newblock {\em The Annals of Mathematics}, 80(2):365--381, 1964.

\bibitem{Wouters09}
J.~Wouters, M.~Fannes, I.~Akhalwaya, and F.~Petruccione.
\newblock Classical capacity of a qubit depolarizing channel with memory.
\newblock {\em Physical Review A}, 79(4):042303, apr 2009.

\bibitem{Zyczkowski_duality_2004}
Karol \.Zyczkowski and Ingemar Bengtsson.
\newblock On duality between quantum maps and quantum states.
\newblock {\em Open Systems \& Information Dynamics}, 11:3--42, March 2004.

\end{thebibliography}

\end{document}